\documentclass[11pt]{article} 

\usepackage[utf8]{inputenc}
\usepackage{geometry} 
\usepackage{graphicx} 

\usepackage{booktabs} 
\usepackage{paralist} 
\usepackage{verbatim} 
\usepackage{subfig} 
\usepackage{caption}
\usepackage{textcomp}
\usepackage{amsmath}
\usepackage{bbm}
\usepackage{bbold}
\usepackage{mathtools}
\usepackage{amsfonts}
\usepackage{amssymb}
\usepackage{amsthm}
\usepackage{algorithm}
\usepackage{algpseudocode}
\usepackage[inline]{enumitem}
\usepackage{sgamevar}
\usepackage[hidelinks]{hyperref}
\usepackage{titling}

\usepackage{fancyhdr} 
\usepackage{sectsty}
\allsectionsfont{\sffamily\mdseries\upshape} 

\usepackage[nottoc,notlof,notlot]{tocbibind} 
\usepackage[titles,subfigure]{tocloft} 

\newcommand{\ts}{\textsuperscript} 
\newtheorem{theorem}{Theorem}

\newtheorem{definition}{Definition}
\newtheorem{fact}{Fact}
\newtheorem{lemma}[theorem]{Lemma}
\newtheorem{prop}[theorem]{Proposition}

\newtheorem{cor}{Corollary}[theorem]

\theoremstyle{remark}
\newtheorem{assumption}[theorem]{Assumption}

\newtheorem*{info-theorem}{Informal Statement of Theorem}
\newtheorem*{info-prop}{Informal Statement of Proposition}
\newcommand{\mc}{\mathcal}

\newcommand{\colspan}[1]{\texttt{ColSpan}(#1)}
\newcommand{\nullspace}[1]{\texttt{null}(#1)}
\newcommand{\rank}[1]{\texttt{rank}({#1})}

\newcommand{\transpose}{\mathsf{T}}
\renewcommand{\vec}[1]{\boldsymbol{\mathbf{#1}}} 
\renewcommand{\Re}{\mathbb{R}}
\newcommand{\vct}[1]{\boldsymbol{\mathbf{#1}}} 
\newcommand{\Rea}{\mathbb{R}}
\newcommand{\Q}{\mathbb{Q}}

\newcommand{\Ce}{\mathbb{C}}
\newcommand{\bigO}{\mathcal{O}}

\title{A Manifold of Polynomial Time Solvable Bimatrix Games}
\author{Joseph L. Heyman\thanks{The author is with the ECE department at OSU. The author was fully supported by the United
States Military Academy and the Army Advanced
Civil Schooling (ACS) program}}
\begin{document}
\begin{titlingpage}
\maketitle
\pagestyle{plain}
\begin{abstract}
This paper identifies a manifold in the space of bimatrix games which contains games that are strategically equivalent to rank-1 games through a positive affine transformation. It also presents an algorithm that can compute, in polynomial time, one such rank-1 game which is strategically equivalent to the original game. Through this approach, we substantially expand the class of games that are solvable in polynomial time. It is hoped that this approach can be further developed in conjunction with other notions of strategic equivalence to compute exact or approximate Nash equilibria in a wide variety of bimatrix games. 
\end{abstract}
\end{titlingpage}

\section{Introduction}\label{sec:Intro}
The study of game theory -- the model of strategic interaction between rational agents -- has wide ranging applicability within the fields of economics, engineering, and computer science. Computing a Nash equilibrium (NE) in a $k$-player finite game is one of the fundamental problems in game theory.\footnote{Nash equilibrium is an acceptable and a widely used solution concept for games; we define it later in the paper.} Due to the well known theorem by Nash in 1951, we know that every finite game has a solution, possibly in mixed strategies \cite{nash1951}. However, computation of a NE has been shown to be hard. For $k\geq3$, \cite{daskalakis2009complexity,daskalakis2005three} proved that computing a NE is Polynomial Parity Argument, Directed Version (PPAD) complete. Indeed, even the $2$-player case has been shown to be PPAD-complete \cite{chen2009}.

In this work we focus on finite, $2$-player bimatrix games in which the payoffs to the players can be represented as two matrices, $A$ and $B$. With the hardness of computing a NE in the bimatrix case well established, identifying classes of games whose solutions can be computed efficiently is an active area of research. As an example in bimatrix games, one can compute the solution of zero-sum games\footnote{A game is called a zero-sum game if the sum of payoffs of the players is zero for all actions of the players.} in polynomial time using the minimax theorem and a simple linear program \cite[p.~152]{von2007theory}. Another subclass of games that admits a polynomial time solution is the class of \textit{strategically zero-sum} games introduced by Moulin and Vial in 1978 \cite{moulin1978strategically}.

More recently, in a series of works \cite{adsul2011rank,adsul2019} the authors have developed polynomial time algorithms for solving another subclass of bimatrix games called rank-$1$ games. Rank-$1$ games are defined as games where the sum of the two payoff matrices is a rank-$1$ matrix. Here we note that this subclass lies in a space of $\Re^{m\times n+m+n}$.

 We say that two games, $G$ and $G^\prime$, are \textit{strategically equivalent} if the optimal strategies of every player in $G$ corresponds to the optimal strategies of every player in $G^\prime$.\footnote{See Section \ref{sec:prelim} for a formal definition and the specific form of strategic equivalence that we study in this work.} Unfortunately, many operations that preserve the strategic equivalence of bimatrix games modify the rank of the game.  For example, the well studied constant-sum game is strategically equivalent to a zero-sum game. 
 However, the zero-sum game has rank zero, while the constant-sum game is a rank-$1$ game. Since the rank of a game influences the most suitable solution technique one should be particularly interested in determining if a given game $G^\prime$ is strategically equivalent to a game $G$, where the rank of $G$ is less than the rank of $G^\prime$.
  
  In this contribution, we do just that. We show that, in polynomial time, it is possible to determine if the game $G^\prime$ lies in a manifold of size $\Re^{m\times n+2m+2n+4}$. If so, we show that the game $G^\prime$ is strategically equivalent to some rank-$1$ game $G$. Furthermore, we show that one can also calculate this rank-$1$ game in polynomial time. As we will show, our algorithm is surprisingly straightforward and significantly expands the space of games which can be solved in polynomial time.

\subsection{Prior Work}\label{subsec:priorWork}
Many papers in the literature have explored the concept of equivalent classes of games. One such concept is \textit{strategic equivalence}, those games that share exactly the same set of NE. 
Indeed, for a classical example, von Neumann and Morgenstern studied strategically equivalent $n$-person zero-sum games \cite[p.~245]{von2007theory} and constant-sum games \cite[p.~346]{von2007theory}.

Closely related to our work is the class of \textit{strategically zero-sum} games defined by Moulin and Vial in \cite{moulin1978strategically}.  
For the bimatrix case, they provide a complete characterization of strategically zero-sum games \cite[Theorem~2]{moulin1978strategically}. While the authors do not analyze the algorithmic implications of their characterization, Kontogiannis and Spirakis do analyze the approach in their work on mutually concave games.  They find that the characterization in \cite[Theorem~2]{moulin1978strategically} can determine whether a bimatrix game is strategically zero-sum, and, if so, calculate an equivalent zero-sum game in time $\bigO{(m^3n^3)}$.

In \cite{kontogiannis2012mutual}, Kontogiannis and Spirakis formulate a quadratic program where the optimal solutions of the quadratic program constitute a subset of the correlated equilibria of a bimatrix game. Furthermore, they then show that this subset of correlated equilibria are exactly the NE of the game. In order to show polynomial tractability of the quadratic program, they define a class of \textit{mutually concave} bimatrix games. Surprisingly, they find that their characterization of \textit{mutually concave} bimatrix games is equivalent to Moulin and Vial's characterization of \textit{strategically zero-sum} bimatrix games \cite[Corollary~2]{kontogiannis2012mutual}.  
They then propose a parameterized version of the Mangasarian and Stone quadratic program that is guaranteed to be convex for a mutually concave game and has time complexity $\bigO{(n^{3.5})}$ for $2\leq m\leq n$ \cite[Theorem~2]{kontogiannis2012mutual}. With the problem of solving a mutually concave game shown to be tractable, the authors then proceed to show that recognizing a mutually concave game can be done in time $\bigO{(m^2n)}$.

As far as we are aware, no other authors have specifically studied games that are strategically equivalent to rank-$1$ games.  With the recently developed fast algorithms \cite{adsul2011rank,adsul2019}  for solving rank-$1$ games, now seems like an appropriate time to do just that.

\subsection{Our Contribution}\label{subsec:ourContribution}
Given a nonzero-sum bimatrix game, $(m,n,\tilde{A},\tilde{B})\in\Q^{m\times n}\times\Q^{m\times n}$, we develop the \textsc{SER1} algorithm that determines whether or not the given game lies in a manifold of size $\Q^{m\times n+2m+2n+4}$. If so, then the game is \textit{strategically equivalent} to some rank-$1$ game $(m,n,A,-A+\vec{r}\vec{c}^\transpose)$. Our approach relies on the classical linear algebra result known as the Wedderburn Rank Reduction Theorem and the theory of matrix pencils. However, we show that it is possible to avoid many of the traditional solution concepts applied to matrix pencils and determine if there exists a solution such that the pencil is a rank-$1$ pencil via solving a single quadratic expression. Thus, our approach is surprisingly simple and leads to an amazingly fast algorithm.  In Section \ref{sec:algSER1} we show that both the determination of strategic equivalence and the computation of the strategically equivalent rank-$1$ game can be done in time $\bigO{(mn+M(\mc{L}))}$, where $\mc{L}$ is the bit-length of the largest absolute value of entries in $(\tilde{A},\tilde{B})$, and $M(\mc{L})$ is the complexity of multiplication. This equivalent rank-$1$ game can then be solved in polynomial time \cite{adsul2011rank,adsul2019}. Since the two games are strategically equivalent, the NE strategies of the equivalent rank-$1$ game are exactly the NE strategies of the original nonzero-sum bimatrix game $(m,n,\tilde{A},\tilde{B})$.

We thus show a significant expansion to the class of bimatrix games that can be solved in polynomial time.

An astute observer whom is familiar with the field would likely recognize that our problem can easily subsume the \textit{strategically zero-sum} case.  However, with multiple existing polynomial time algorithms for that case \cite{moulin1978strategically,kontogiannis2012mutual,heyman2018SERO}, we choose to specifically focus on strategically equivalent rank-$1$ games.

\subsection{Notation}\label{subsec:notation}
We use $\vct{1}_n$ and $\vct{0}_n$ to denote, respectively, the all ones and all zeros vectors of length $n$. All vectors are annotated by bold font, e.g $\vct{u}$, and all vectors are treated as column vectors. $\Delta_n$ is the set of probability distributions over $ \{1,\ldots,n\}$, where $\Delta_{n}=\big\{\mathbf{p} \mid p_i\geq0,\forall i\in \{1,\dots,n\},\sum_{i=1}^n p_i = 1  \big\}$.
Let $\vct{e}_j$, $j \in \{1,2,...,n\}$, denote the vector with $1$ in the $j$\ts{th} position and $0$'s elsewhere. 

Consider a matrix $C$. We use $\rank{C}$ to indicate the rank of the matrix $C$. $\colspan{C}$ indicates the subspace spanned by the columns of the matrix $C$, also known as the column space of the matrix $C$. We indicate the nullspace of the matrix $C$, the space containing all solutions to $Cx=\vct{0}_m$, as
$\nullspace{C}$. In addition, we use $C^{(j)}$ to denote the j\ts{th} column of $C$ and $C_{(i)}$ to denote the i\ts{th} row of $C$. 

\section{Preliminaries}\label{sec:prelim}
In this section, we recall some basic definitions in bimatrix games and the definition of strategic equivalence in bimatrix games.

We consider here a two player game, in which player 1 (the row player) has $m$ actions and player 2 (the column player) has $n$ actions. Player 1's set of pure strategies is denoted by $S_1=\{1,\dots,m\}$ and player 2's set of pure strategies is $S_2=\{1,\dots,n\}$. If the players play pure strategies $(i,j)\in S_1 \times S_2$, then player 1 receives a payoff of $a_{i,j}$ and player 2 receives $b_{i,j}$.

We let $A=[a_{i,j}]\in \Rea^{m\times n}$ represent the payoff matrix of player 1 and $B=[b_{i,j}]\in \Rea^{m\times n}$ represent the payoff matrix of player 2. As the two-player finite game can be represented by two matrices, this game is commonly referred to as a bimatrix game. The bimatrix game is then defined by the tuple $(m,n,A,B)$. Define the $m\times n$ matrix $C$ as the sum of the two payoff matrices, $C \coloneqq A+B$. We define the rank of a game as $\rank{C}$.\footnote{Some authors define the rank of the game to be the maximum of the rank of the two matrices $A$ and $B$, but this is not the case here.}  

Players may also play mixed strategies, which correspond to a probability distribution over the available set of pure strategies. Player 1 has mixed strategies $\vct{p}$ and player 2 has mixed strategies $\vct{q}$, where $\vct{p}\in \Delta_m$ and $\vct{q}\in \Delta_n$. Using the notation introduced above, player 1 has expected payoff $\vct{p}^\transpose A \vct{q}$ and player 2 has expected payoff $\vct{p}^\transpose B\vct{q}$. 

A Nash Equilibrium is defined as a tuple of strategies $(\vct{p}^*,\vct{q}^*)$ such that each player's strategy is an optimal response to the other player's strategy. In other words, neither player can benefit, in expectation, by unilaterally deviating from the Nash Equilibrium. This is made precise in the following definition.
\begin{definition}[Nash Equilibrium \cite{nash1951}]\label{def:NE}
	We refer to the pair of strategies $(\vct{p}^*,\vct{q}^*)$ as a Nash Equilibrium (NE) if and only if:
	\begin{equation*}
	\vct{p}^{*\transpose }A\vct{q^*}\geq \vct{p}^{\transpose}A\vct{q^*} \quad
	\forall \; \vct{p}\in\Delta_m;\quad
	\vct{p}^{*\transpose }B\vct{q^*}\geq \vct{p}^{*\transpose}B\vct{q} \ \; \quad 
	\forall \; \vct{q}\in\Delta_n.
	\end{equation*}
\end{definition}
It is a well known fact due to Nash \cite{nash1951} that every bimatrix game with finite set of pure strategies has at least one NE in mixed strategies. However, one can define games in which multiple NE exist in mixed strategies. Let $\Phi:\Rea^{m\times n}\times \Rea^{m\times n}\rightrightarrows \Delta_m\times\Delta_n$ be the Nash equilibrium correspondence\footnote{A correspondence is a set valued map \cite[p. 555]{ali2006}.}: Given the matrices $(A,B)$, $\Phi(A,B)\subset  \Delta_m\times\Delta_n$ denotes the set of all Nash equilibria of the game $(m,n,A,B)$. Note that due to the result in \cite{nash1951}, $\Phi(A,B)$ is nonempty for every $(A,B)\in \Rea^{m\times n}\times \Rea^{m\times n}$.

We say that two games are strategically equivalent if both games have the same set of players, the same set of strategies per player, and the same set of Nash equilibria. The following definition formalizes this concept.
\begin{definition} \label{def:stratEqNE}
	The 2-player finite games $(m,n,A,B)$ and $(m,n,\tilde{A},\tilde{B})$ are strategically equivalent iff $\Phi(A,B) = \Phi(\tilde A, \tilde B)$.
\end{definition}

We now have a well known Lemma on strategic equivalence in bimatrix games that is typically stated without proof.\footnote{See e.g \cite{moulin1978strategically,kontogiannis2012mutual}.} For completeness, the proof is in Appendix \ref{app:randomProofs}. 

\begin{lemma}\label{lem:stratEqVec}
	Let $A,B\in\Rea^{m\times n}$ be two matrices.
	Let $\tilde{A} = \alpha_1A+\beta_1\vct{1}_m \vct{u}^\transpose$ and $\tilde{B} = \alpha_2B+\beta_2\vct{v}\vct{1}_n^\transpose$ where $\beta_1,\beta_2 \in \Rea$, $\alpha_1,\alpha_2 \in \Rea_{>0}$, $\vct{u}\in\Rea^n$, and $\vct{v}\in\Rea^m$. Then the game $(m,n,A,B)$ is strategically equivalent to $(m,n,\tilde{A},\tilde{B})$.
\end{lemma}

\section{Problem Formulation and Main Results}\label{sec:probForm}
From Lemma \ref{lem:stratEqVec}, we conclude that if there exists $\beta_1,\beta_2 \in \Re$, $\alpha_1,\alpha_2 \in \Re_{>0}$, $\vct{u}\in\Re^n$, and $\vct{v}\in\Re^m$ such that:
\begin{align}
    \tilde{A} &= \alpha_1A+\beta_1\vct{1}_m \vct{u}^\transpose, \label{eq:AtildeRank1}\\
    \tilde{B} &= -\alpha_2A+\alpha_2\vec{r}\vec{c}^\transpose+\beta_2\vct{v}\vct{1}_n^\transpose \label{eq:BtildeRank1}
\end{align}
then $(m,n,\tilde{A},\tilde{B})$ is strategically equivalent  to the rank-$1$ game $(m,n,A,-A+\vec{r}\vec{c}^\transpose)$ via a positive affine transformation (PAT). Here we note that there may be other transforms, even nonlinear transforms, which maintain the set of NE for a specific game. We don't consider those transforms and focus on the class of PATs.  Throughout this work, any mention of strategic equivalence refers to strategic equivalence via a PAT.

Combining \eqref{eq:AtildeRank1} and \eqref{eq:BtildeRank1}, we have:
\begin{align}
		\tilde{A} &= -\frac{\alpha_1}{\alpha_2} \tilde{B}+\alpha_1\vec{r}\vec{c}^\transpose+ \beta_1\vct{1}_m \vct{u}^\transpose+\frac{\alpha_1}{\alpha_2}\beta_2\vct{v}\vct{1}_n^\transpose. \label{eq:AtildeRank1Combined}
\end{align}
    
Defining $\gamma\coloneqq\frac{\alpha_1}{\alpha_2}$, $D\coloneqq\beta_1\vct{1}_m \vct{u}^\transpose+\gamma\beta_2\vct{v}\vct{1}_n^\transpose$, and $\vec{\hat{r}}\vec{\hat{c}}^\transpose\coloneqq\alpha_1\vec{r}\vec{c}^\transpose$, we rewrite \eqref{eq:AtildeRank1Combined} as:
\begin{equation}\label{eq:defineD}
    \tilde{A}+\gamma\tilde{B}=\beta_1\vct{1}_m \vct{u}^\transpose+\gamma\beta_2\vct{v}\vct{1}_n^\transpose+\alpha_1\vec{r}\vec{c}^\transpose=D+\vec{\hat{r}}\vec{\hat{c}}^\transpose.
\end{equation}

In addition, let us define $\tilde{C}(\gamma)\coloneqq \tilde{A}+\gamma\tilde{B}$ and the subspace $\mc{M}_{m\times n}(\Re)$, which we define as
\begin{equation*}
    \mc{M}_{m\times n}(\Re)=\Big\{M\in\Rea^{m\times n}\Big\vert \text{ there exists $\vct{u}\in\Re^n$, and $\vct{v}\in\Re^m$ such that } M=\vct{1}_m\vct{u}^\transpose+\vct{v}\vct{1}_n^\transpose\Big\}.
\end{equation*}
From \eqref{eq:defineD}, it is clear that if $(m,n,\tilde{A},\tilde{B})$ is strategically equivalent to the rank-$1$ game $(m,n,A,-A+\vec{r}\vec{c}^\transpose)$ via a PAT, then there exists a $\gamma\in \Re_{>0}$ such that $\rank{\tilde{C}(\gamma)}\leq3$. However, while $\rank{\tilde{C}(\gamma)}\leq3$ is a necessary condition, it is not sufficient.  For sufficiency, we also require the existence of $D\in\mc{M}_{m\times n}(\Re)$, $\vec{\hat{r}}\in\Re^m$, and $\vec{\hat{c}}\in\Re^n$ such that $\tilde{C}(\gamma)=D+\vec{\hat{r}}\vec{\hat{c}}^\transpose$.\footnote{By the definition of $D$, we have $0\leq\rank{D}\leq2$. This implies that $1\leq\rank{\tilde{C}(\gamma)}\leq3$.} 

Thus, what we have shown above is the following result:
\begin{prop}\label{prop:forward}
$(m,n,\tilde{A},\tilde{B})$ is strategically equivalent to $(m,n,A,-A+\vec{r}\vec{c}^\transpose)$ through a PAT if and only if $\tilde{C}(\gamma)=D+\vec{\hat{r}}\vec{\hat{c}}^\transpose$, where $D\in\mc{M}_{m\times n}(\Re)$.
\end{prop}
\begin{proof}
The proof follows from the preceding discussions.
\end{proof}

In what follows, we show the converse holds. Furthermore, in Section \ref{sec:proofMainResult} we devise methods for determining $\gamma$, $D$, and $\tilde{C}(\gamma)$. Finally, in Section \ref{sec:algSER1} we construct an algorithm that efficiently implements the results from Section \ref{sec:proofMainResult}.

For a matrix $M$, decompose it in a manner such that $ M = D+K$, where $D\in\mc{M}_{m\times n}$ and $K\not\in\mc{M}_{m\times n}$. Let $g(M)$ denote the matrix $K$. This function can be computed efficiently.
\begin{assumption}\label{assm:converse}
The game $(m,n,\tilde{A},\tilde{B})$ satisfies
\begin{enumerate}
    \item $D\in\mc{M}_{m\times n}$  \; $\forall\gamma\in\Re_{>0}$. 
    \item $g(\tilde C(\gamma^*))$ is a rank-1 matrix.
\end{enumerate}
\end{assumption}

\begin{theorem}
If Assumption \ref{assm:converse} holds, then there exists a matrix $\hat A\in\Rea^{m\times n}$ and vectors $\vec{\hat{r}}\in\Re^m$, $\vec{\hat{c}}\in\Re^n$ such that the bimatrix game $(m,n,\tilde{A},\tilde{B})$ is strategically equivalent to the rank-$1$ game $(m,n,\hat{A},-\hat{A}+\vec{\hat{r}}\vec{\hat{c}}^\transpose)$.
\end{theorem}
\begin{proof}
This result is a direct consequence of Theorem \ref{thm:rank1} presented in Subsection \ref{subsec:proofMainResult:wedderburn} and Theorem \ref{thm:rank1Less3} presented in Appendix \ref{app:stratEqRank1}.
\end{proof}

We briefly note here that, even if one were to be given $A,\vec{r},\vec{c}$, and a game $(m,n,\tilde{A},\tilde{B})$ that is strategically equivalent to $(m,n,A,-A+\vec{r}\vec{c}^\transpose)$ it would be impossible to exactly recover $A,\vec{r},\vec{c}$ without also knowing $\beta_1,\beta_2,\alpha_1,\alpha_2,\vct{u}$, and $\vct{v}$. Thus, our technique calculates another game, $(m,n,\hat{A},-\hat{A}+\vec{\hat{r}}\vec{\hat{c}}^\transpose)$ which is rank-$1$ and strategically equivalent to both $(m,n,\tilde{A},\tilde{B})$ and $(m,n,A,-A+\vec{r}\vec{c}^\transpose)$. 

Let us now turn our attention to showing the necessary conditions under which Assumption \ref{assm:converse} holds true and devising methods to test those conditions.

\section{Proof of the Main Result}\label{sec:proofMainResult}

In this section, we prove the main theoretical results presented in Section \ref{sec:probForm}.
Let us now turn our attention to decomposing the game via the Wedderburn Rank Reduction Formula.

\subsection{Reducing the Rank of a Game via the Wedderburn Rank Reduction Formula}\label{subsec:proofMainResult:wedderburn}
In this subsection, we show that if there exists a $\gamma\in\Re_{>0}$ such that $\rank{\tilde{C}}=3$ where $\tilde{C}:=\tilde{C}(\gamma)=\tilde{A}+\gamma\tilde{B}$, then it may be possible, under certain conditions, to calculate a rank-1 game that is strategically equivalent to the rank-k game $(m,n,\tilde{A},\tilde{B})$.
In what follows, we first present the Wedderburn Rank Reduction formula upon which our technique is based, and then we proceed to state our result, leaving the proof to Appendix \ref{app:subsec:rankC3}.  

The Wedderburn Rank Reduction formula is a classical technique in linear algebra that allows one to reduce the rank of a matrix by subtracting a specifically formulated rank-1 matrix. By repeated applications of the formula, one can obtain a matrix decomposition as the sum of multiple rank-1 matrices.
\footnote{For further reading on the Wedderburn rank reduction formula, we refer the reader to Wedderburn's original book \cite[p.~69]{wedderburn1934lectures}, or to the excellent treatment of the topic by Chu et al. \cite{chu1995rank}.}

We now proceed to state Wedderburn's original theorem. Following that, we show how one can exploit the flexibility of the decomposition to extract specifically formulated rank-1 matrices that allow us, when certain conditions hold true, to reduce the rank of a bimatrix game.

\begin{theorem}[{\cite[p.~69]{wedderburn1934lectures} \cite{chu1995rank}}] \label{thm:wedderburn}
	Let $C\in\Re^{m\times n}$ be an arbitrary matrix, not identically zero. Then, there exists vectors $\vec{x}_1\in \Re^n$ and $\vec{y}_1\in \Re^m$ such that $w_1=\vec{y}_1^\transpose C \vec{x}_1\neq0$. Then, setting $C=C_1$ for convenience, the matrix
	\begin{equation} \label{eq:Wedderburn}
	C_2 \coloneqq C_1-w_1^{-1}C_1\vec{x}_1\vec{y}_1^\transpose C_1
	\end{equation}
	has rank exactly one less than the rank of $C$.
\end{theorem}
\begin{proof}
The original proof of \eqref{eq:Wedderburn} is due to Wedderburn \cite[p.~69]{wedderburn1934lectures}. See Appendix \ref{app:linAlg} for our restated version.
\end{proof}

Let us now proceed to show how one can apply Theorem \ref{thm:wedderburn} in order to reduce the rank of a bimatrix game. For ease of exposition, we break the result into two theorems--Theorem \ref{thm:rank1} and Theorem \ref{thm:rank1Less3} (presented in Appendix \ref{app:stratEqRank1}). 

\begin{theorem}\label{thm:rank1}
	Consider the game  $(m,n,\tilde{A},\tilde{B})$. If there exists $\gamma\in \Re_{>0}$ such that $\tilde{C}:=\tilde{C}(\gamma)=\tilde{A}+\gamma\tilde{B}=D+\hat{\vec{r}}\hat{\vec{c}}^\transpose$ where $D\in\mc{M}_{m\times n}(\Re)$ with $\rank{D}=2$, $\rank{\tilde{C}}=3$, $\hat{\vec{r}}\in\Re^m$, $\hat{\vec{c}}\in\Re^n$, then there exists $\vec{x}_1,\vec{x}_2\in\Rea^n$ and $\vec{y}_1,\vec{y}_2\in\Rea^m$ such that:
	\begin{enumerate}
	    \item $\tilde{C}\vec{x}_1=\vec{1}_m$ and $w_1=\vct{y}_1^\transpose \tilde{C}\vct{x}_1\neq 0$. Let $\hat{\vec{u}}^\transpose=w_1^{-1}\vec{y}_1^\transpose\tilde{C}$ and compute $\tilde{C}_2$ using \eqref{eq:Wedderburn} as follows:
	\begin{equation*}
    	\tilde{C}_2 = \tilde{C}-w_1^{-1}\tilde{C}\vct{x}_1\vct{y}_1^\transpose \tilde{C}=\tilde{C}-\vct{1}_m\hat{\vct{u}}^\transpose. 
	\end{equation*}
	\item $\vec{y}_2\tilde{C}_2=\vec{1}_n^\transpose$ and $w_2=\vct{y}_2^\transpose \tilde{C}_2\vct{x}_2\neq 0$. Let $\hat{\vec{v}}=w_2^{-1}\tilde{C}_2\vec{x}_2$ and compute $\tilde{C}_3$ using \eqref{eq:Wedderburn} as follows:
	\begin{equation*}
    	\tilde{C}_3 = \tilde{C}_2-w_2^{-1}\tilde{C}_2\vct{x}_2\vct{y}_2^\transpose \tilde{C}_2=\tilde{C}-\hat{\vct{v}}\vct{1}_n^\transpose=\tilde{C}-\vct{1}_m\hat{\vct{u}}^\transpose-\hat{\vct{v}}\vct{1}_n^\transpose. 
	\end{equation*}
	\end{enumerate}
	
	Define $\hat A = \tilde{A} - \vec{1}_m\hat{\vec{u}}^\transpose$ and $\hat B = \gamma \tilde{B}-\hat{\vct{v}}\vct{1}_n^\transpose$. Then, $\hat A + \hat{B} = \tilde{C}_3=\hat{\vec{r}}\hat{\vec{c}}^\transpose$ and the bimatrix game $(m,n,\tilde{A},\tilde{B})$ is strategically equivalent to the rank-$1$ game $(m,n,\hat{A},\hat{B})$.
\end{theorem}

\begin{proof}
    See Appendix \ref{app:stratEqRank1}.
\end{proof}

In Theorem \ref{thm:rank1} we assumed the existence of $\gamma$ such that $\tilde{C}$ is in our desired form.  In the next subsection, we show an efficient decomposition of $\tilde{C}$ that we can conduct without knowing $\gamma$. Following that, we then show a technique for solving for such a $\gamma$ (if it exists).

\subsection{An Efficient Decomposition of \texorpdfstring{$\tilde{C}$}{}} 
First-off, let us assume that there exists a $\gamma$ such that $\tilde{C}(\gamma)=M+D$, where  $D\in\mc{M}_{m\times n}(\Re)$ with $\rank{D}=2$. See Appendix \ref{app:decompC} for the cases of $\rank{D}<2$. Therefore, we know that there exists a vector $\vec{x}_1$ such that $\tilde{C}\vct{x}_1=\vct{1}_m$ and a $\vec{y}_2$ such that $\vec{y}_2\tilde{C}_2=\vec{1}_n^\transpose$. Examining Theorem \ref{thm:rank1}, we see that the term $\vec{x}_1$ always appears as a term in the matrix-vector product $\tilde{C}\vct{x}_1$. Therefore, we can replace all appearances of the term $\tilde{C}\vct{x}_1$ with $\vct{1}_m$ and there is no need to explicitly calculate $\vec{x}_1$.
 
An algorithmic efficient manner for selecting $\vec{y}_1$ is to let $\vec{y}_1=\vec{e}_i$ for some $i$. Then we have that:
\begin{equation*}
    w_1^{-1}=\frac{1}{\vct{y}_1^\transpose \tilde{C}\vct{x}_1}=\frac{1}{\vct{e}_i^\transpose \vec{1}_m}=1.
\end{equation*}
Then, we also have:
\begin{align*}
    \vct{1}_m\hat{\vct{u}}^\transpose&=w_1^{-1}\tilde{C}\vct{x}_1\vct{y}_1^\transpose \tilde{C}=\vec{1}_m\vec{e}_i^\transpose \tilde{C}=\vec{1}_m\vec{e}_i^\transpose (\tilde{A}+\gamma\tilde{B})\\
    &=(\vec{1}_m\tilde{A}_{(i)}+\gamma\vec{1}_m\tilde{B}_{(i)}).
\end{align*}
Then we have that $\tilde{C}_2=\tilde{A}+\gamma\tilde{B}-(\vec{1}_m\tilde{A}_{(i)}+\gamma\vec{1}_m\tilde{B}_{(i)})$. Then, letting $\vec{y}_2\tilde{C}_2=\vec{1}_n^\transpose$ and $\vec{x}_2=\vec{e}_j$, we have $w_2^{-1}=1$ and
\begin{align*}
    \hat{\vct{v}}\vec{1}_n^\transpose&=\tilde{C}_2\vec{e}_j\vec{1}_n^\transpose
    =[\tilde{A}+\gamma\tilde{B}-(\vec{1}_m\tilde{A}_{(i)}+\gamma\vec{1}_m\tilde{B}_{(i)})]\vec{e}_j\vec{1}_n^\transpose\\
    &=(\tilde{A}^{(j)}-\vec{1}_m\tilde{a}_{i,j})\vec{1}_n^\transpose+\gamma(\tilde{B}^{(j)}-\vec{1}_m\tilde{b}_{i,j})\vec{1}_n^\transpose.
\end{align*}
Therefore, we have 
\begin{equation}\label{eq:C3LongForm}
   \tilde{C}_3=\tilde{A}+\gamma\tilde{B}-(\vec{1}_m\tilde{A}_{(i)}+\gamma\vec{1}_m\tilde{B}_{(i)})-\big((\tilde{A}^{(j)}-\vec{1}_m\tilde{a}_{i,j})\vec{1}_n^\transpose+\gamma(\tilde{B}^{(j)}-\vec{1}_m\tilde{b}_{i,j})\vec{1}_n^\transpose\big). 
\end{equation}
For notational simplicity, let $D(\gamma)\coloneqq(\vec{1}_m\tilde{A}_{(i)}+\gamma\vec{1}_m\tilde{B}_{(i)})-\big((\tilde{A}^{(j)}-\vec{1}_m\tilde{a}_{i,j})\vec{1}_n^\transpose+\gamma(\tilde{B}^{(j)}-\vec{1}_m\tilde{b}_{i,j})\vec{1}_n^\transpose\big)$ and $\bar{A}+\gamma\bar{B}\coloneqq\tilde{C}_3$. Clearly, we have $D(\gamma)\in\mc{M}_{m\times n}(\Re)$ for all $\gamma\in\Re$.

\begin{theorem}\label{thm:rank1GammaUnknown}
    	Consider the game $(m,n,\tilde{A},\tilde{B})$, select any $(i,j)\in\{1\dots m\} \times \{1\dots n\}$, and formulate $\bar{A}+\gamma\bar{B}\coloneqq\tilde{C}_3$ as in \eqref{eq:C3LongForm}. Let $\hat{A}(\gamma)=\tilde{A}-(\vec{1}_m\tilde{A}_{(i)}+\gamma\vec{1}_m\tilde{B}_{(i)})$ and $\hat{B}(\gamma)=\gamma\tilde{B}-\big((\tilde{A}^{(j)}-\vec{1}_m\tilde{a}_{i,j})\vec{1}_n^\transpose+\gamma(\tilde{B}^{(j)}-\vec{1}_m\tilde{b}_{i,j})\vec{1}_n^\transpose\big)$. If there exists $\gamma^*\in\Re_{>0}$ such that $\rank{\bar{A}+\gamma^*\bar{B}}=1$, then the game $(m,n,\tilde{A},\tilde{B})$ is strategically equivalent to the rank-$1$ game $(m,n,\hat{A}(\gamma^*),\hat{B}(\gamma^*))$.   
\end{theorem}
\begin{proof}
    The proof follows from the preceding discussion and Theorem \ref{thm:rank1}.
\end{proof}
Thus, what we've shown is that our problem is equivalent to forming a rather simple matrix decomposition to calculate $D(\gamma)$, and then determining whether there exists a scalar parameter $\gamma^*$ such that $\rank{\bar{A}+\gamma^*\bar{B}}=1$. Let us now turn our attention to the theory of \textit{matrix pencils} and develop an efficient procedure for calculating $\gamma^*$.   

\subsection{Rank-1 Matrix Pencils}\label{subsec:rank1Pencils} 
We would be remiss if we did not mention that fact that problems in the form $A+\lambda B$, with $A,B\in\Re^{m\times n}$ matrices of known values and $\lambda\in\Ce$ an unknown scalar parameter, are typically referred to as \textit{linear matrix pencils} (or just  \textit{pencils})\cite[p.~24]{gantmakher1959theory},\cite{ikramov1993matrix}. Therefore, our problem of finding a $\gamma^*\in\Re$ such that $\rank{\tilde{C}(\gamma^*)}=1$ is intimately connected to the theory of matrix pencils and the \textit{generalized eigenvalues} of matrix pencils.  However, we eschew the traditional techniques for solving matrix pencils in order to avoid the higher than necessary computational complexity and possible numerical difficulties.  Instead, in this subsection, we develop a series of results that allow us to determine whether or not there exists a $\lambda^*\in\Ce$ such that $\rank{A+\lambda^*B}=1$ by solving for the roots of a single polynomial equation (at worse a quadratic) and then conducting at most two matrix comparisons.

Let us begin by stating some facts about rank-1 matrices that will allow us to easily ascertain when a given matrix is rank-$1$ and to solve for values of $\lambda^*$, when they exist, such that the matrix pencil, $A+\lambda^* B$, is a rank-$1$ pencil.

\begin{fact}\label{fact:rank1facts}
The matrix $M$ in $\Re^{m\times n}$ is rank-$1$ if and only if $M\neq\vec{0}_{m\times n}$ and the following hold true:
\begin{enumerate}
    \item Every row (column) of $M$ is a scalar multiple of every other row (column) of $M$.
    \item Choose any element $m_{i,j}$ of $M$ such that $m_{ij}\neq0$ and form $\vec{r}_j=M^{(j)}$, $\vec{c}_i^\transpose=m_{i,j}^{-1}M_{(i)}$. Then, $M=\vec{r}_j\vec{c}_i^\transpose$.
\end{enumerate}
\end{fact}

\begin{theorem}\label{thm:rank1Pencil}
Given the matrix pencil, $A+\lambda B$, with $A,B\in\Re^{m \times n}$, assume that $A,B\neq\vec{0}_{m\times n}$ and $\rank{A+B}>1$. Construct $\lambda$-vectors $\vec{r}_j(\lambda)$ and $\vec{c}_i(\lambda)$ as follows:
\begin{enumerate}
    \item Choose any $(i,j)\in\{1\dots n\}\times \{1\dots m\}$ such that $a_{i,j}\neq0$. Such an $a_{i,j}$ is guaranteed to exist since $A\neq\vec{0}_{m\times n}$.
    \item Let $\vec{r}_j(\lambda)=A^{(j)}+\lambda B^{(j)}$.
    \item Let $\vec{c}_i(\lambda)^\transpose=\frac{1}{a_{i,j}+\lambda b_{i,j}}(A_{(i)}+\lambda B_{(i)})$.
\end{enumerate}
Then there exists $\lambda^*\in\Ce$ such that $\rank{A+\lambda^* B}=1$ if and only if either:
\begin{enumerate}
    \item $b_{i,j}\neq0$ and $\rank{A+\frac{-a_{i,j}}{b_{i,j}} B}=1$; or \label{cond:thm:rank1Pencil:bij}
    \item $A+\lambda^*B=\vec{r}_j(\lambda^*)\vec{c}_i(\lambda^*)^\transpose$. \label{cond:thm:rank1Pencil:decomp}
\end{enumerate}
\end{theorem}
\begin{proof}

    We first prove the forward direction. Suppose there exists $\lambda^*\in\Ce$ such that $\rank{A+\lambda^* B}=1$. We split the proof of the forward direction into two cases:
    
    {\bf Case $\lambda^*=\frac{-a_{i,j}}{b_{i,j}}$:} Suppose that $\lambda^*=\frac{-a_{i,j}}{b_{i,j}}$. Since $\lambda^*\in\Ce$, this implies $b_{i,j}\neq0$. Furthermore, $\rank{A+\lambda^* B}=1$ and $\lambda^*=\frac{-a_{i,j}}{b_{i,j}}$ implies $\rank{A+\frac{-a_{i,j}}{b_{i,j}} B}=1$. In addition, note that $\lambda^*=\frac{-a_{i,j}}{b_{i,j}}$ implies that $\vec{c}_i(\lambda^*)$ is undefined; therefore, the expression $A+\lambda^*B=\vec{r}_j(\lambda^*)\vec{c}_i(\lambda^*)^\transpose$ is undefined and cannot hold true.
    
    {\bf Case $\lambda^*\neq \frac{-a_{i,j}}{b_{i,j}}$:} Now, suppose $\lambda^*\neq\frac{-a_{i,j}}{b_{i,j}}$. Then $\rank{A+\frac{-a_{i,j}}{b_{i,j}} B}\neq1$. Also, $a_{i,j}+\lambda^*b_{i,j}\neq0$, so $\vec{c}_i(\lambda^*)$ is well-defined. Then $A+\lambda^*B=\vec{r}_j(\lambda^*)\vec{c}_i(\lambda^*)^\transpose$ follows from Fact \ref{fact:rank1facts}.
    
    We next prove the reverse direction. Suppose $b_{i,j}\neq0$ and $\rank{A+\frac{-a_{i,j}}{b_{i,j}} B}=1$. Then $\lambda^*=\frac{-a_{i,j}}{b_{i,j}}\in\Re\subset\Ce$ and $\rank{A+\lambda^* B}=1$. Of course, as in above, since $\vec{c}_i(\lambda^*)^\transpose$ is undefined at $\lambda^*=\frac{-a_{i,j}}{b_{i,j}}$ by definition, we conclude that $A+\lambda^*B\neq \vec{r}_j(\lambda^*)\vec{c}_i(\lambda^*)^\transpose$.
    
    Now, suppose $A+\lambda^*B=\vec{r}_j(\lambda^*)\vec{c}_i(\lambda^*)^\transpose$, which implies that $\vec{c}_i(\lambda^*)$ is well-defined. This implies that $b_{i,j}=0$ and/or $\rank{A+\frac{-a_{i,j}}{b_{i,j}} B}\neq1$. Furthermore, since $\lambda^*$ is the solution to a system of $m\times n$ linear or quadratic equations with real coefficients, we have $\lambda^*\in\Ce$. Then,  $\rank{A+\lambda^* B}=1$ follows from Fact \ref{fact:rank1facts}.
\end{proof}

Since it is rather trivial to determine whether or not $b_{i,j}\neq0$ and $\rank{A+\frac{-a_{i,j}}{b_{i,j}} B}=1$, we'll assume throughout the rest of this subsection that $\rank{A+\frac{-a_{i,j}}{b_{i,j}} B}\neq1$. 
Let us now examine the matrix equality $A+\lambda B=\vec{r}_j(\lambda)\vec{c}_i(\lambda)^\transpose$, introduce some additional notation, and state some lemmas that allow us to determine whether or not there exists a finite $\lambda^*$ such that $\rank{A+\lambda^* B}=1$.

With $\vec{r}_j(\lambda)$ and $\vec{c}_i(\lambda)$ as defined in Theorem \ref{thm:rank1Pencil}, let us write the following system of equations:
\begin{align}\label{eq:rank1Pencil:system}
    A+\lambda B&=\vec{r}_j(\lambda)\vec{c}_i(\lambda)^\transpose\\
    \begin{bmatrix}
    a_{1,1}+\lambda b_{1,1} &  \dots & a_{1,n}+\lambda b_{1,n} \\
    \vdots                 & \ddots  & \vdots                          \\
    a_{m,1}+\lambda b_{m,1}  & \dots & a_{m,n}+\lambda b_{m,n}
    \end{bmatrix}
    &=
    \begin{bmatrix}\label{eq:rank1Pencil:decompMatrix}
    \frac{(a_{1,j}+\lambda b_{1,j})(a_{i,1}+\lambda b_{i,1})}{a_{i,j}+\lambda b_{i,j}} & \dots & \frac{(a_{1,j}+\lambda b_{1,j})(a_{i,n}+\lambda b_{i,n})}{a_{i,j}+\lambda b_{i,j}} \\
    \vdots                           & \ddots& \vdots                          \\
    \frac{(a_{m,j}+\lambda b_{m,j})(a_{i,1}+\lambda b_{i,1})}{a_{i,j}+\lambda b_{i,j}} &  \dots & \frac{(a_{m,j}+\lambda b_{m,j})(a_{i,n}+\lambda b_{i,n})}{a_{i,j}+\lambda b_{i,j}}
    \end{bmatrix}
\end{align}
Since $\lambda$ is a scalar variable, it is clear from \eqref{eq:rank1Pencil:decompMatrix} that \eqref{eq:rank1Pencil:system} only has a solution (or possibly multiple solutions), $\lambda^*$, if $\lambda^*$ simultaneously satisfies $m\times n$ single-variable polynomials, where each polynomial is of degree at most $2$. Thus, one could solve all $m\times n$ single-variable polynomials and then check whether or not every solution has a common value. While this procedure is somewhat efficient, we'll show below that at most $(m-1)\times(n-1)$ of the polynomials have finite solutions and therefore contribute any meaningful information.  In addition, we'll show that it is sufficient to identify one polynomial that is not the zero polynomial and then conduct a matrix checking problem for the solution(s) of that polynomial.

Let us now introduce notation for the polynomial represented by row $s$ and column $t$ in \eqref{eq:rank1Pencil:decompMatrix}. For $(s,t)\in\{1\dots n\}\times \{1\dots m\}$, let 
\begin{equation}\label{eq:fst}
f_{s,t}(i,j;\lambda)=a_{s,t}+\lambda b_{s,t}-\frac{(a_{s,j}+\lambda b_{s,j})(a_{i,t}+\lambda b_{i,t})}{a_{i,j}+\lambda b_{i,j}}
\end{equation}

From \eqref{eq:fst}, it is clear that when $s=i$ or $t=j$, then $f_{i,t}(i,j;\lambda)$ and $f_{s,j}(i,j;\lambda)$ are the zero polynomial.  In other words, $f_{s,t}(i,j;\hat{\lambda})=0$ trivially holds true for all $\hat{\lambda}\in\Ce$ for one entire column and one entire row of \eqref{eq:rank1Pencil:decompMatrix}. Since these $m+n-1$ expressions hold true for all $\hat{\lambda}\in\Ce$, they lend no information for determining whether or not there exists $\lambda^*$ such that $\rank{A+\lambda^* B}=1$. Thus, we can disregard these $m+n-1$ polynomials and only consider the remaining $(m-1)\times(n-1)$ polynomials. Let us now show that at least one of the remaining $(m-1)\times(n-1)$ polynomials is not the zero polynomial. The proof of the Lemma is deferred to Appendix \ref{app:randomProofs}.    

\begin{lemma}\label{lem:existNotZeroPoly}
Let  $A,B\in\Re^{m \times n}$, and assume that $A,B\neq\vec{0}_{m\times n}$. Given the matrix pencil, $A+\lambda B$, construct $\lambda$-vectors $\vec{r}_j(\lambda)$ and $\vec{c}_i(\lambda)$ as in Theorem \ref{thm:rank1Pencil} and $f_{s,t}(i,j;\lambda)$ as in \eqref{eq:fst}. 
If $\rank{A+B}>1$, then there exists at least one pair $(s,t)\in\{1\dots n\}\times \{1\dots m\}$ such that $f_{s,t}(i,j;\lambda)$ is not the zero polynomial.    
\end{lemma}

\begin{theorem}\label{thm:rank1PencilSolveLambda}
    Let  $A,B\in\Re^{m \times n}$, and assume that $A,B\neq\vec{0}_{m\times n}$ with $\rank{A+B}>1$. Given the matrix pencil, $A+\lambda B$, construct $\lambda$-vectors $\vec{r}_j(\lambda)$ and $\vec{c}_i(\lambda)$ as in Theorem \ref{thm:rank1Pencil}. Pick any $(l,k)\in\{1\dots n\}\times \{1\dots m\}$ such that $f_{l,k}(i,j;\lambda)$ is not the zero polynomial. Let $\hat{\lambda}_1,\hat{\lambda}_2$ be solutions to $f_{l,k}(i,j;\lambda)=0$. Then, there exists $\lambda^*\in\Ce$ such that $\rank{A+\lambda^*B}=1$ if and only if $A+\hat{\lambda}_1 B=\vec{r}_j(\hat{\lambda}_1)\vec{c}_i(\hat{\lambda}_1)^\transpose$ or/and $A+\hat{\lambda}_2 B=\vec{r}_j(\hat{\lambda}_2)\vec{c}_i(\hat{\lambda}_2)^\transpose$. 
\end{theorem}
\begin{proof}
    Let us first note that, by Lemma \ref{lem:existNotZeroPoly}, there exists at least one pair $(l,k)\in\{1\dots n\}\times \{1\dots m\}$ such that $f_{l,k}(i,j;\lambda)$ is not the zero polynomial.
    
    Now, suppose there exists $\lambda^*\in\Ce$ such that $\rank{A+\lambda^* B}=1$. Then, by Theorem \ref{thm:rank1Pencil}, $A+\lambda^*B=\vec{r}_j(\lambda^*)\vec{c}_i(\lambda^*)^\transpose$. Thus, $f_{s,t}(i,j;\lambda^*)=0$ for all $(s,t)\in\{1\dots n\}\times \{1\dots m\}$. In particular, $f_{l,k}(i,j;\lambda^*)=0$. Therefore, either $\hat{\lambda}_1=\lambda^*$ or/and $\hat{\lambda}_2=\lambda^*$. It then follows that $A+\hat{\lambda}_1B=\vec{r}_j(\hat{\lambda}_1)\vec{c}_i(\hat{\lambda}_1)^\transpose$ or/and $A+\hat{\lambda}_2 B=\vec{r}_j(\hat{\lambda}_2)\vec{c}_i(\hat{\lambda}_2)^\transpose$.
    
    In the other direction, let $\hat{\lambda}_1,\hat{\lambda}_2$ be solutions to $f_{l,k}(i,j;\lambda)=0$. Note that $f_{l,k}(i,j;\lambda)=0$ may be a linear equation.  In that case, for simplicity, let $\hat{\lambda}_1=\hat{\lambda}_2$.    
    Now, suppose $A+\hat{\lambda}_1B=\vec{r}_j(\hat{\lambda}_1)\vec{c}_i(\hat{\lambda}_1)^\transpose$ and let $\lambda^*=\lambda_1$. Then $A+\lambda^*B=\vec{r}_j(\lambda^*)\vec{c}_i(\lambda^*)^\transpose$ and $\rank{A+\lambda^*B}=1$ by Theorem \ref{thm:rank1Pencil}.  The case of $A+\hat{\lambda}_2 B=\vec{r}_j(\hat{\lambda}_2)\vec{c}_i(\hat{\lambda}_2)^\transpose$ is similar and therefore omitted.  
\end{proof}
Let us use $\Lambda$ to represent the solution set obtained from Theorem \ref{thm:rank1PencilSolveLambda}. Note, $\Lambda$ has a maximum cardinality of $2$ and may, of course, be empty.

Throughout the presentation in this subsection, we've assumed $\rank{A+B}>1$. For completeness, we present the next trivial lemma without proof.
\begin{lemma}\label{lem:trivialLambda}
If $\rank{A+B}=1$, then $\lambda^*=1$ and $\rank{A+\lambda^*B}=1$. 
\end{lemma}

\subsection{Solving for \texorpdfstring{$\gamma^*$}{}} 
Let us now present our main result for solving for $\gamma^*$ and calculating the equivalent rank-$1$ game.
\begin{theorem}
	Consider the game $(m,n,\tilde{A},\tilde{B})$ and formulate $\bar{A}+\lambda\bar{B}$, $\hat{A}(\lambda)$, and $\hat{B}(\lambda)$ as in Theorem \ref{thm:rank1GammaUnknown}. Construct the set $\Lambda$ as follows: If $\rank{\bar{A}+\bar{B}}=1$, let $\Lambda=\{1\}$.  Otherwise, construct $\lambda$-vectors $\vec{r}_j(\lambda)$ and $\vec{c}_i(\lambda)$ as in Theorem \ref{thm:rank1Pencil}. If $b_{i,j}\neq0$ and $\rank{A+\frac{-a_{i,j}}{b_{i,j}} B}=1$, let $\Lambda=\{\frac{-a_{i,j}}{b_{i,j}}\}$. Otherwise, apply Theorem \ref{thm:rank1PencilSolveLambda} to calculate the solution set $\Lambda$. 
	
	If there exists $\lambda\in\Lambda$ such that $\lambda\in\Re_{>0}$, then let $\gamma^*$ equal such a $\lambda$. Then the game $(m,n,\tilde{A},\tilde{B})$ is strategically equivalent to the rank-$1$ game $(m,n,\hat{A}(\gamma^*),\hat{B}(\gamma^*))$. 
\end{theorem}
\begin{proof}
    The proof follows from Theorems \ref{thm:rank1GammaUnknown}, \ref{thm:rank1Pencil}, \ref{thm:rank1PencilSolveLambda}, and Lemma \ref{lem:trivialLambda}.
\end{proof}

\section{Algorithm for Strategically Equivalent Rank-1 Games (SER1)}\label{sec:algSER1}
We have shown that given the game $(m,n,\tilde{A},\tilde{B})$, one can determine if the game is strategically equivalent to the game $(m,n,A,-A+\vec{r}\vec{c}^\transpose)$ through a PAT. If so, then it is possible to construct a rank-$1$ game which is strategically equivalent to the original game. One can then efficiently solve the strategically equivalent rank-$1$ game via existing polynomial time algorithms \cite{adsul2011rank,adsul2019}.  We state the key steps in Algorithm \ref{alg:solveRank1Short} and show that both the determination of strategic equivalence and the computation of the strategically equivalent rank-$1$ game can be done in time $\bigO{(mn+M(\mc{L}))}$, where $\mc{L}$ is the bit-length of the largest absolute value of entries in $(\tilde{A},\tilde{B})$, and $M(\mc{L})$ is the complexity of multiplication.

The analytical results and discussions throughout this paper apply to real bimatrix games, with $(m,n,\tilde{A},\tilde{B})\in\Re^{m\times n} \times \Re^{m\times n}$. However, for computational reasons, when discussing the algorithmic implementations we focus on rational bimatrix games, with  $(m,n,\tilde{A},\tilde{B})\in\Q^{m\times n} \times \Q^{m\times n}$.

We present here a shortened version of the SER1 algorithm that applies when $\rank{D}=2$. The cases of $\rank{D}<2$ are similar, but add complexity to the presentation.  See Appendix \ref{app:algAnalysis} for a discussion of extending \textsc{ShortSER1} to handle these cases.

\begin{theorem}\label{thm:compTimeSER1}
The \textsc{ShortSER1} algorithm determines if a game $(m,n,\tilde{A},\tilde{B})$ is strategically equivalent to a rank-$1$ game and returns the strategically equivalent rank-$1$ game in time $\bigO{(mn+M(\mc{L}))}$.
\end{theorem}
\begin{proof}
Calculating $\bar{A},\bar{B}$ is simply a series of vector outer products and matrix subtractions. Therefore, it takes time $\bigO{(mn)}$. By Fact \ref{fact:rank1facts}, determining whether a matrix is rank-$1$ or not is equivalent to $mn$ divisions and $mn$ comparisons, and therefore $\bigO{(mn)}$. As mentioned in Subsection \ref{subsec:rank1Pencils}, finding an $(s,t)\in\{1\dots m\} \times \{1\dots n\}$ such that $f_{s,t}(i,j;\lambda)$ is not the zero polynomial requires a search over a space of $(m-1)\times(n-1)$, thus time $\bigO{(mn)}$. Calculating the coefficients of $f_{s,t}(i,j;\lambda)$ requires at most $4$ scalar multiplications and $3$ scalar additions/subtractions, which requires time $\bigO{(1)}$. Here, we note that the bit length of the coefficients is thus $2\mc{L}+1$, which is $\bigO{(\mc{L})}$. Now, at worst solving $f_{s,t}(i,j;\lambda)=0$ is equivalent to computing a square root, which has time complexity $M(\mc{L})$, where $M(\mc{L})$ is the complexity of the chosen algorithm for multiplying two $\mc{L}$ bit numbers \cite{alt1979square}. Furthermore, we note that only one such $f_{s,t}(i,j;\lambda)=0$ must be solved. With candidate values of $\hat{\lambda}$ thus determined, checking whether $\bar{A}+\hat{\lambda}\bar{B}=\vec{r}_j(\hat{\lambda})\vec{c}_i(\hat{\lambda})^\transpose$ requires one vector outer product and one matrix comparison, thus takes time $\bigO{(mn)}$. Finally, calculating $\hat{A}$ and $\hat{B}$ requires a series of vector outer products, scalar-matrix multiplications and matrix additions/subtractions, which can be done in time $\bigO{(mn)}$. Thus, overall, the algorithm takes time $\bigO{(mn+M(\mc{L}))}$.
\end{proof}

Of course, after calculating the equivalent rank-$1$ game, one would likely wish to calculate an NE of that game.  The currently known algorithms for calculating an NE of a rank-$1$ game require at least one call to a linear program solver \cite{adsul2011rank,adsul2019}. At the time of this writing, solving a linear program takes more time than $\bigO{(mn+M(\mc{L}))}$. Thus, the overall running time of solving an NE for a strategically equivalent rank-$1$ game is not dominated by our algorithm.

\section{Conclusion} \label{sec:conclusion}
This paper identifies a class of games that can be reduced to a rank-$1$ game in polynomial time. This yields a polynomial time algorithm to determine a Nash equilibrium of a large class of nonzero-sum games, which may include full rank games. Consequently, we have identified a manifold in the space of games that are solvable in polynomial time.

For the future, we intend to analyze nonzero-sum games for which one can derive a rank-$1$ game that approximates the original game. This can pave way for new methods to compute approximate Nash equilibria in bimatrix games.

\begin{algorithm}
	\caption{Condensed algorithm for identifying a strategically equivalent rank-1 game}
	\label{alg:solveRank1Short}
	\begin{algorithmic}[1] 
		\Procedure{ShortSER1}{$\tilde{A},\tilde{B}$}
		\State $\Lambda\gets\emptyset$
		\State choose $(l,k)\in\{1\dots m\} \times \{1\dots n\}$
		\State $\bar{A}\gets\tilde{A}-\vec{1}_m\tilde{A}_{(l)}-(\tilde{A}^{(k)}-\vec{1}_m\tilde{a}_{l,k})\vec{1}_n^\transpose$
		\State $\bar{B}\gets\tilde{B}-\vec{1}_m\tilde{B}_{(l)}-(\tilde{B}^{(k)}-\vec{1}_m\tilde{b}_{l,k})\vec{1}_n^\transpose$
		\If{$\rank{\bar{A}+\bar{B}}=1$} 
		\State $\Lambda\cup\{1\}$
		\Else
		    \State choose $(i,j)\in\{1\dots m\} \times \{1\dots n\}$ s.t. $\bar{a}_{i,j}\neq0$
		    \If{$\bar{b}_{i,j}\neq0$}
		        \If{$\rank{\bar{A}+\frac{-\bar{a}_{i,j}}{\bar{b}_{i,j}}\bar{B}}=1$}
		            \State $\Lambda\cup\{\frac{-\bar{a}_{i,j}}{\bar{b}_{i,j}}\}$
		        \EndIf
		    \EndIf
		    \If{$\Lambda=\emptyset$}
		        \State $\vec{r}_j(\lambda)\gets\bar{A}^{(j)}+\lambda\bar{B}^{(j)}$ ; $\vec{c}_i(\lambda)\gets\frac{1}{\bar{a}_{i,j}+\lambda\bar{b}_{i,j}}(\bar{A}_{(i)}+\lambda\bar{B}_{(i)})$
		        \State choose $(s,t)\in\{1\dots m\} \times \{1\dots n\}$ s.t. $f_{s,t}(i,j;\lambda)$ is not the zero polynomial
		        \State $\Gamma\gets\text{solve}(f_{s,t}(i,j;\lambda)=0)$
		        \For{$\hat{\lambda}\in\Gamma$}
		            \If{$\bar{A}+\hat{\lambda}\bar{B}=\vec{r}_j(\hat{\lambda})\vec{c}_i(\hat{\lambda})^\transpose$}
		            \State $\Lambda\cup\{\hat{\lambda}\}$
		            \EndIf
		        \EndFor
		    \EndIf
		    \If{$(\Lambda=\emptyset)\lor(\nexists\lambda\in\Lambda$ s.t $\lambda\in\Q_{>0}$)}
		        \State Not strategically equivalent via PAT. \textbf{exit} 
		    \Else
		        \State $\gamma^*\gets\lambda\in\Q_{>0}$
		      \State $\hat{A}=\tilde{A}-(\vec{1}_m\tilde{A}_{(l)}+\gamma^*\vec{1}_m\tilde{B}_{(l)})$
		      \State $\hat{B}=\tilde{B}-\big((\tilde{A}^{(k)}-\vec{1}_m\tilde{a}_{l,k})\vec{1}_n^\transpose+\gamma^*(\tilde{B}^{(k)}-\vec{1}_m\tilde{b}_{l,k})\vec{1}_n^\transpose\big)$
		        \EndIf
		      \EndIf  
		\EndProcedure
	\end{algorithmic}
\end{algorithm}

\newpage
\appendix
\section{Proofs of Auxiliary Lemmas} \label{app:randomProofs}
\begin{proof}[Proof of Lemma \ref{lem:stratEqVec}]
	Since $\vct{p}\in\Delta_m$ and $\vct{q}\in\Delta_m$, we have $\vct{p}^\transpose \vct{1}_m=1$ and $\vct{1}_n^\transpose \vct{q}=1$. Then, $\forall (\vct{p},\vct{q})\in(\Delta_m\times\Delta_n)$ we have:
	\begin{align*}
	\vct{p}^\transpose \tilde{A}\vct{q}&=\vct{p}^\transpose(\alpha_1A+\beta_1\vct{1}_m \vct{u}^\transpose)\vct{q}=\alpha_1\vct{p}^\transpose A\vct{q}+\beta_1\vct{u}^\transpose \vct{q},\\
	\vct{p}^\transpose \tilde{B}\vct{q}&=\vct{p}^\transpose(\alpha_2B+\beta_2\vct{v}\vct{1}_n^\transpose)\vct{q}=\alpha_2\vct{p}^\transpose B\vct{q}+\beta_2\vct{p}^\transpose \vct{v}.
	\end{align*}
	Now, assume that $(\vct{p}^*,\vct{q}^*)$ is an NE of $(m,n,\tilde{A},\tilde{B})$. Then, for player 1, 
	\begin{align*}
	\vct{p}^{*\transpose} \tilde{A}\vct{q}^*&=\alpha_1\vct{p}^{*\transpose} A\vct{q}^*+\beta_1\vct{u}^\transpose \vct{q}^* \\ 
	&\geq \alpha_1\vct{p}^{\transpose} A\vct{q}^*+\beta_1\vct{u}^\transpose \vct{q}^*=\vct{p}^{\transpose} \tilde{A}\vct{q}^* \; \,
	\forall \vct{p}\in\Delta_m\\
	&\iff \vct{p}^{*\transpose }A\vct{q^*}\geq \vct{p}^{\transpose}A\vct{q^*} \quad
	\forall \vct{p}\in\Delta_m.
	\end{align*}
	Similarly, for player 2, 
	\begin{align*}
	\vct{p}^{*\transpose} \tilde{B}\vct{q}^*&=\alpha_2\vct{p}^{*\transpose} B\vct{q}^*+\beta_2\vct{p}^{*\transpose} \vct{v} \\
	&\geq \alpha_2\vct{p}^{*\transpose}B\vct{q}+\beta_2\vct{p}^{*\transpose} \vct{v} = \vct{p}^{*\transpose} \tilde{B}\vct{q} \; \, \forall \vct{q}\in\Delta_n\\
	&\iff \vct{p}^{*\transpose }B\vct{q^*}\geq \vct{p}^{*\transpose}B\vct{q} \quad
	\forall \vct{q}\in\Delta_n.
	\end{align*}
	Then Definition \ref{def:stratEqNE} is satisfied, and $(\vct{p}^*,\vct{q}^*)$ is an NE of $(m,n,A,B)$ if and only if $(\vct{p}^*,\vct{q}^*)$ is an NE of $(m,n,\tilde{A},\tilde{B})$.
\end{proof}

\begin{proof}[Proof of Lemma \ref{lem:existNotZeroPoly}]
Suppose, by way of contradiction, that $f_{s,t}(i,j;\lambda)$ is the zero polynomial for all $(s,t)\in\{1\dots n\}\times \{1\dots m\}$. This implies that for any $\hat{\lambda}\in\Ce$ and for all $(s,t)\in\{1\dots n\}\times \{1\dots m\}$, $f_{s,t}(i,j;\hat{\lambda})=0$. Furthermore, this implies that for any $\hat{\lambda}\in\Ce$,  $A+\hat{\lambda} B=\vec{r}_j(\hat{\lambda})\vec{c}_i(\hat{\lambda})^\transpose$ and $\rank{A+\hat{\lambda} B}=1$.  In particular, $\rank{A+B}=1$, which is a contradiction.     
\end{proof}

\section{Results From Linear Algebra}\label{app:linAlg}
In this section, we collect some preliminary results from linear algebra that are used throughout the rest of the paper. In the first subsection, we present the proof of the Wedderburn rank reduction formula and some related results. The second subsection discusses certain properties of the vector space $\mc{M}_{m\times n}$ that are used in the proofs of results in the subsequent sections.

\subsection{Wedderburn Rank Reduction Formula}

\begin{proof}[Proof of Theorem \ref{thm:wedderburn}]
The original proof of \eqref{eq:Wedderburn} is due to Wedderburn \cite[p.~69]{wedderburn1934lectures}. We restate it here for completeness. 
	
	We first show that the null space of $C_2$ contains the null space of $C_1$. Pick $\vct z$ such that $C_1\vct{z}=\vct{0}$. Then,
	\begin{equation*}
	C_2\vct{z}=C_1\vct{z}-w_1^{-1}C_1\vct{x}_1\vct{y}_1^\transpose C_1\vct{z}=\vct{0}.
	\end{equation*} 
	Thus, $\vct z$ is in the null space of $C_2$, which implies that the null space of $C_2$ contains the null space of $C_1$. 
	
	Next, we show that $\vct x_1$ is in the null space of $C_2$, thereby showing that the dimension of the null space of $C_2$ is one more than dimension of the null space of $C_1$ (since $C_1 \vct x_1 \neq \vct{0}$). Consider
	\begin{equation*}
	C_2\vct{x}_1=C_1\vct{x}_1-w_1^{-1}C_1\vct{x}_1\vct{y}_1^\transpose C_1\vct{x}_1=\vct{0}.
	\end{equation*}
	Thus, the rank of $C_2$ is one less than the rank of $C_1$.
	
\end{proof}	
We now have the following theorem that applies the Wedderburn rank reduction formula to compute a decomposition of the matrix.

\begin{theorem}[{Rank-Reducing Process \cite[p.~69]{wedderburn1934lectures} \cite{chu1995rank}}] \label{thm:rankReduce}
	Let $C \in \Re^{m\times n}$. If $\rank{C}=\gamma$, then there exists $\vec{x}_k\in\Re^n$, $\vec{y}_k\in\Re^m$, $k = 1,\ldots,\gamma$ such that $w_k=\vec{y}_k^\transpose C_k \vec{x}_k\neq0$ and the following holds:
	\begin{equation}\label{eq:thm:rankReduce:recursive}
	C_{k+1}=C_k-w_k^{-1}C_k\vec{x}_k\vec{y}_k^\transpose C_k \quad k=(1,2,\dots,\gamma)
	\end{equation}
	where $C_1=C$, $C_{\gamma+1}=0$, and $\rank{C_{k+1}}=\rank{C_k}-1$. In addition, let $W_k=w_k^{-1}C_k\vec{x}_k\vec{y}_k^\transpose C_k$, where $\rank{W_k}=1$. Then, we have
	\begin{equation}\label{eq:thm:rankReduce:W}
		C=\sum_{k=1}^{\gamma} W_k.
	\end{equation}
\end{theorem}
\begin{proof}
	Apply Theorem \ref{thm:wedderburn} to $C$ for $\gamma+1$ iterations.
\end{proof}

\begin{cor}\label{cor:thm:rankReduce}
For a matrix $M\in\Re^{m\times n}$ with $\rank{M}=r$, let $\{W_k\}$ for $k=(1,\dots,r)$ be a set of rank-1 matrices derived from the rank-reducing process in Theorem \ref{thm:rankReduce}. Let $v_ku_k^\transpose=W_k=w_k^{-1}M_k\vec{x}_k\vec{y}_k^\transpose M_k$. Then $\vec{v}_k\notin\colspan{M_{j}}$ and $\vec{u}_k\notin\colspan{M_{j}^\transpose}$ for all $j>k$. 
\end{cor}
\begin{proof}
By recursively applying \eqref{eq:thm:rankReduce:recursive} from Theorem \ref{thm:rankReduce}, we can write $M_{k+1}$ as 
\begin{equation*}
    M_{k+1}=M-\sum_{i=1}^{k}W_i.
\end{equation*}
Then, combined with \eqref{eq:thm:rankReduce:W} we have
\begin{align}
    M_{k+1}&=\sum_{i=1}^{r}W_i-\sum_{i=1}^{k}W_i\label{eq:cor:thm:rankReduce}\\
    M_{k+1}&=\sum_{i=k+1}^{r}W_i
\end{align}
From \eqref{eq:cor:thm:rankReduce} and Fact \ref{fact:rankFactorization}, it is apparent that the basis for $\colspan{M_{k+1}}$ is formed by removing $[\vec{v}_1,\vec{v}_2,\dots,\vec{v}_k]$ from the basis of $\colspan{M}$. Then, since by definition the basis vectors are linearly independent, there is no linear combination of $[\vec{v}_j,\vec{v}_{j+1},\dots,\vec{v}_r]$, such that the linear combination equals $\vec{v}_k$ for $j>k$. Therefore, $\vec{v}_k\notin\colspan{M_{j}}$ and for all $j>k$. The same argument holds for $\vec{u}_k\notin\colspan{M_{j}^\transpose}$ for all $j>k$.   
\end{proof}

Motivated by \cite{chu1995rank}, we now have the following proposition.

\begin{prop}\label{prop:colSpanC2}
We use here the same notation as in Theorem \ref{thm:wedderburn}. Let $\rank{C}=k$, with $k \geq 2 $. Let $\{\vec{x}_1,\vec{y}_1\}$ be vectors associated with a rank-reducing process (so that $\vec{y}_1^\transpose C \vec{x}_1\neq0$). We have $\vec{z}\in \colspan{C_2^\transpose}$ if and only if $\vec{z}\in \colspan{C^\transpose}$ and $\vec{z}\perp\vec{x}_1$.
\end{prop}

\begin{proof}

Suppose that $\vec{z}\in \colspan{C_2^\transpose}$. Then there exists a $\vec{y}_2\in \Re^m$ such that $C_2^\transpose \vec{y}_2=\vec{z}$. Choose such a $\vec{y}_2$ and define $\vec{v}_2$ as
\begin{equation}\label{eq:v2Def}
    \vec{v}_2\coloneqq \vec{y}_2-
	\frac{\vec{y}_2^\transpose C \vec{x}_1}{\vec{y}_1^\transpose C \vec{x}_1}\vec{y}_1. 
\end{equation}

Directly from Theorem \ref{thm:wedderburn}, we have
	\begin{align}
	\vec{y}_2^\transpose C_2&= \vec{y}_2^\transpose C-w_1^{-1}\vec{y}_2^\transpose C\vec{x}_1\vec{y}_1^\transpose C \nonumber \\
    \vec{y}_2^\transpose C_2&=\big(\vec{y}_2^\transpose- \frac{\vec{y}_2^\transpose C \vec{x}_1}{\vec{y}_1^\transpose C \vec{x}_1}\vec{y}_1^\transpose \big )C \nonumber \\
    \vec{y}_2^\transpose C_2&= \vec{v}_2^\transpose C \label{eq:prop:colSpanC2:v2C}
	\end{align}

Then by \eqref{eq:prop:colSpanC2:v2C}, $\vec{z}=C^\transpose\vec{v}_2$ which implies that $\vec{z}\in \colspan{C^\transpose}$. From the proof of Theorem \ref{thm:wedderburn}, it is clear that $\vec{x}_1\in \nullspace{C_2}$. Then
\begin{equation}\label{eq:prop:colSpanC2:v2Perpx1}
    \vec{y}_2^\transpose C_2\vec{x}_1=\vec{v}_2^\transpose C \vec{x}_1=0.
\end{equation}

Finally, with  $\vec{z}^\transpose=\vec{v}_2^\transpose C$ we have that $\vec{z}\perp\vec{x}_1$ by \eqref{eq:prop:colSpanC2:v2Perpx1}.

Now, suppose that $\vec{z}\in \colspan{C^\transpose}$ and $\vec{z}\perp\vec{x}_1$. Since $\rank{C}\geq2$, the $\colspan{C^\transpose}$ is at least a two dimensional subspace, and thus there exists a $\vec{y}_1\in \Re^m$ such that $C^\transpose \vec{y}_1$ and $\vec{z}$ are linearly independent. Choose such a $\vec{y}_1$. Choose such a $\vec{y}_1$. Then, $\vec{y}_1^\transpose C\vec{x}_1\neq0$ and $\vec{v}_2$, as given in \eqref{eq:v2Def}, is well-defined. Select $\vec{y}_2$ such that $C^\transpose \vec{y}_2=\vec{z}$. Then,
\begin{equation}\label{eq:prop:colSpanC2:v2Equaly2}
    \vec{v}_2=\vec{y}_2-
	\frac{\vec{z}^\transpose\vec{x}_1}{\vec{y}_1^\transpose C \vec{x}_1}\vec{y}_1
	=\vec{y}_2-\frac{0}{\vec{y}_1^\transpose C \vec{x}_1}\vec{y}_1=\vec{y}_2.
\end{equation}

    By \eqref{eq:prop:colSpanC2:v2C} and \eqref{eq:prop:colSpanC2:v2Equaly2} , we then have:
\begin{equation*}
    \vec{y}_2^\transpose C_2=\vec{v}_2^\transpose C=\vec{y}_2^\transpose C=\vec{z}.
\end{equation*}
Therefore, $\vec{z}\in \colspan{C_2^\transpose}$. 
\end{proof}

\subsection{The Subspace \texorpdfstring{$\mc{M}_{m\times n}(\Re)$}{}}\label{subsec:subspaceM}
In this subsection, we recall an essential fact from linear algebra and introduce the the subspace $\mc{M}_{m\times n}(\Re)$, which we define as 
\[\mc{M}_{m\times n}(\Re)=\Big\{M\in\Rea^{m\times n}\vert M=\vct{1}_m\vct{u}^\transpose+\vct{v}\vct{1}_n^\transpose\Big\}.\]
The properties of this subspace are essential to formulating the algorithmic solution to our problem as the matrix $D$ that we are searching for must lie in this subspace.

Let us begin by stating an essential fact on the rank-1 decomposition of a matrix.
\begin{fact}\label{fact:rankFactorization}
For any matrix $M\in\Re^{m\times n}$ with $\rank{M}=r$, one can write $M$ as a summation of $r$ rank-1 matrices, $M=\sum_{1}^{r}\vec{v}_i\vec{u}_i^\transpose$. Furthermore, $[\vec{v}_1,\vec{v}_2,\dots,\vec{v}_r]$ is a basis for $\colspan{M}$ and $[\vec{u}_1,\vec{u}_2,\dots,\vec{u}_r]$ is a basis for $\colspan{M^\transpose}$.
\end{fact}

We now proceed to prove that $\mc{M}_{m\times n}(\Rea)$ is indeed a subspace of the vector space of all real matrices.  Following that, we state some essential properties of the subspace that will be used throughout the presentation.
\begin{lemma} \label{lem:MsubSpace}
Let $\mc{V}_{m\times n}(\Re)$ be the vector space of real $m\times n$ matrices and $\mc{M}_{m\times n}(\Re)$ be the space of real matrices such that for all $M\in\mc{M}_{m\times n}(\Re)$, $M=\vct{1}_m\vct{u}^\transpose+\vct{v}\vct{1}_n^\transpose$. Then $\mc{M}_{m\times n}(\Re)$ is a subspace of $\mc{V}_{m\times n}(\Re)$.
\end{lemma}
\begin{proof}
We show that $\mc{M}_{m\times n}(\Re)$ meets the three properties of a subspace of a vector space.  First, consider $\vec{u}=\vec{0}_n$,$\vec{v}=\vec{0}_m$. Then, $M=\vec{0}_{m\times n}$ and the zero vector is in $\mc{M}_{m\times n}(\Re)$.  Secondly, for all $M_1,M_2\in\mc{M}_{m\times n}(\Re)$, 
\begin{align*}
    M_1+M_2&=\vct{1}_m\vct{u}^\transpose+\vct{v}\vct{1}_n^\transpose+\vct{1}_m\vct{u}_2^\transpose+\vct{v}_2\vct{1}_n^\transpose,\\
    &=\vct{1}_m(\vct{u}^\transpose+\vct{u}_2^\transpose)+(\vct{v}+\vct{v}_2)\vct{1}_n^\transpose.
\end{align*}
Therefore, $M_1+M_2\in\mc{M}_{m\times n}(\Re)$. Finally, for all $c\in\Re$ and $M\in\mc{M}_{m\times n}(\Re)$, $cM=\vct{1}_m c\vct{u}^\transpose+c\vct{v}\vct{1}_n^\transpose\in\mc{M}_{m\times n}(\Re)$.  
\end{proof}
\begin{lemma} \label{lem:existXY}
For any matrix $M\in\Re^{m\times n}$ with $M\neq\vec{0}_{m\times n}$ there is at least one nonzero column and one nonzero row. Let $M^{(j)}$ and $M_{(i)}$ denote such a column and row. Then, with $\vec{x}=\vec{e}_j$ and $\vec{y}=\vec{e}_i$ we have that $\vec{1}_n^\transpose \vec{x}\neq0$,$\vec{1}_m^\transpose \vec{y}\neq0$, $M\vec{x}=M^{(j)}\neq\vec{0}_m$, and $\vec{y}^\transpose M=M_{(i)}\neq\vec{0}_n^\transpose$. 
\end{lemma}
\begin{proof}
The proof is straightforward and therefore omitted. 
\end{proof}

\begin{lemma}\label{lem:Mplusrc}
For any matrices $C\in\Re^{(m\times n)}$ and $M\in\mc{M}_{m\times n}(\Re)$ with $\rank{C}=c$, $\rank{M}=m$ such that $C=M+\sum_{1}^{c-m}\vec{r}_i\vec{c}_i^\transpose$:
\begin{enumerate}
    \item If $\rank{M}=2$, then $\vec{1}_m\in\colspan{C}$ and $\vec{1}_n\in\colspan{C^\transpose}$. In addition, for all $\vec{x},\vec{y}$ such that $C\vec{x}=\vec{1}_m$,$C^\transpose \vec{y}=\vec{1}_n$, we have $\vec{1}_n^\transpose \vec{x}=0$,$\vec{1}_m^\transpose \vec{y}=0$.\label{itm:lem:MplusrcRank2}
    \item If $\rank{M}=1$, then either $\vec{1}_m\in\colspan{C}$, or $\vec{1}_n\in\colspan{C^\transpose}$ or both $\vec{1}_m\in\colspan{C}$ and $\vec{1}_n\in\colspan{C^\transpose}$.\label{itm:lem:MplusrcRank1}
\end{enumerate}
\end{lemma}
\begin{proof}
For claim \ref{itm:lem:MplusrcRank2}, $\vec{1}_m\in\colspan{C}$ and $\vec{1}_n\in\colspan{C^\transpose}$ follows directly from $\rank{M}=2$ and Fact \ref{fact:rankFactorization}. Then, for all $\vec{x}$ such that $C\vec{x}=\vec{1}_m$ we have that:
\begin{align*}
    C\vec{x}&=\vec{1}_m\vec{u}^\transpose \vec{x}+\vec{v}\vec{1}_n^\transpose\vec{x}+\sum_{1}^{c-m}\vec{r}_i\vec{c}_i^\transpose \vec{x}=\vec{1}_m,\\
    &=(\vec{u}^\transpose \vec{x})\vec{1}_m+(\vec{1}_n^\transpose\vec{x})\vec{v}+\sum_{1}^{c-m}(\vec{c}_i^\transpose \vec{x})\vec{r}_i=\vec{1}_m,\\
    &\implies(\vec{1}_n^\transpose\vec{x})\vec{v}=(1-\vec{u}^\transpose \vec{x})\vec{1}_m+\sum_{1}^{c-m}(\vec{c}_i^\transpose \vec{x})\vec{r}_i.
\end{align*}
From Fact \ref{fact:rankFactorization} we have that $\vec{v},\vec{r}_i,\vec{1}_m$ must be linearly independent for all $i$. This implies that the equation above is satisfied if and only if $\vec{1}_n^\transpose\vec{x}=0$, $\vec{u}^\transpose \vec{x}=1$, and $\vec{c}_i^\transpose \vec{x}=0$ for all $i$. To prove that for all $\vec{y}$ such that $C^\transpose \vec{y}=\vec{1}_n$, $\vec{1}_m^\transpose \vec{y}=0$ apply the same technique to $C^\transpose$.

Claim \ref{itm:lem:MplusrcRank1} follows directly from $\rank{M}=1$ and Fact \ref{fact:rankFactorization}.
\end{proof}

\section{Strategic Equivalence between Bimatrix Games and Rank-1 Games}\label{app:stratEqRank1}
In this section, we present the proofs of our main results. 
 For ease of exposition, we break the result into two theorems--Theorem \ref{thm:rank1} and Theorem \ref{thm:rank1Less3}--proved in the next two subsections. 

\subsection{The case of \texorpdfstring{$\rank{\tilde{C}}$}{} equals 3}\label{app:subsec:rankC3}

\begin{proof}[Proof of Theorem \ref{thm:rank1}]
Lemma \ref{lem:Mplusrc}, $D\in\mc{M}_{m\times n}(\Re)$, and $\rank{D}=2$ implies that $\vct{1}_m\in \colspan{\tilde{C}}$, $\vct{1}_n\in\colspan{\tilde{C}^\transpose}$. Therefore, there exists $\vct{x}_1$ such that $\tilde{C}\vct{x}_1=\vct{1}_m$.  Lemma \ref{lem:existXY} implies that there exists $\vct y_1\in\Rea^m$ such that $\vct{y}_1^\transpose \tilde{C}\neq \vct 0_n$ and $\vec{1}_m^\transpose \vct{y}_1\neq 0$. Therefore, we have $w_1=\vct{y}_1^\transpose \tilde{C}\vct{x}_1\neq0$. Let $w_1^{-1}\vct{y}_1^\transpose \tilde{C} \coloneqq \hat{\vct{u}}^\transpose$, then $w_1^{-1}\tilde{C}\vct{x}_1\vct{y}_1^\transpose \tilde{C}=\vct{1}_m\hat{\vct{u}}^\transpose$. After applying the Wedderburn rank reduction formula, we have
	\begin{equation*}
    	\tilde{C}_2= \tilde{C}-w_1^{-1}\tilde{C}\vct{x}_1\vct{y}_1^\transpose \tilde{C}=\tilde{C}-\vct{1}_m\hat{\vct{u}}^\transpose.
	\end{equation*}
	
We now proceed to show that there exists $\vec{y}_2$ such that $\vec{y}_2\tilde{C}_2=\vec{1}_n^\transpose$. We've already shown that $\vct{1}_n\in\colspan{\tilde{C}^\transpose}$. By Lemma \ref{lem:Mplusrc}, we have that $\vec{1}_n^\transpose \vec{x}_1=0$. Then, by Proposition \ref{prop:colSpanC2}, $\vct{1}_n\in\colspan{\tilde{C}^\transpose}$ and $\vec{1}_n^\transpose \vec{x}_1=0$ implies that $\vct{1}_n\in\colspan{\tilde{C}_2^\transpose}$. Therefore, there exists $\vec{y}_2$ such that $\vec{y}_2^\transpose\tilde{C}_2=\vec{1}_n^\transpose$.   

Let us now show the existence of $\vec{x}_2$ such that $w_2\neq0$. By Theorem \ref{thm:wedderburn} and the assumption that $\rank{\tilde{C}}=3$, we have that $\rank{\tilde{C}_2}=2$, which implies that $\tilde{C}_2\neq\vec{0}_{m\times n}$. Therefore, by Lemma \ref{lem:existXY} there exists $\vec{x}_2$ such that $\tilde{C}_2\vec{x}_2\neq\vec{0}_m$ and $\vec{1}_n^\transpose\vec{x}_2\neq0$. Then, $w_2=\vec{y}_2^\transpose\tilde{C}_2\vec{x}_2\neq0$. Let $\hat{\vec{v}}=w_2^{-1}\tilde{C}_2\vec{x}_2$, then $w_2^{-1}\tilde{C}_2\vct{x}_2\vct{y}_2^\transpose \tilde{C}_2=\hat{\vct{v}}\vct{1}_n^\transpose$. Again, we apply the Wedderburn rank reduction formula to obtain 
	\begin{equation*}
    	\tilde{C}_3 = \tilde{C}_2-w_2^{-1}\tilde{C}_2\vct{x}_2\vct{y}_2^\transpose \tilde{C}_2=\tilde{C}-\hat{\vct{v}}\vct{1}_n^\transpose=\tilde{C}-\vct{1}_m\hat{\vct{u}}^\transpose-\hat{\vct{v}}\vct{1}_n^\transpose. 
	\end{equation*}

The assumption that $\rank{\tilde{C}}=3$ and Theorem \ref{thm:rankReduce} implies that $\rank{\tilde{C}_3}=1$. Furthermore, the above construction and Corollary \ref{cor:thm:rankReduce} implies that $\vec{1}_m\notin\colspan{\tilde{C}_3}$ and $\vec{1}_n\notin\colspan{\tilde{C}_3^\transpose}$.  Therefore, $\tilde{C}_3=\hat{\vec{r}}\hat{\vec{c}}^\transpose$.  

Finally, defining $\hat A \coloneqq \tilde{A} - \vec{1}_m\hat{\vec{u}}^\transpose$ and $\hat B \coloneqq \gamma \tilde{B}-\hat{\vct{v}}\vct{1}_n^\transpose$ we have that
\begin{align*}
	\hat{A}+\hat{B}&=\tilde{A}+\gamma\tilde{B}-\vct{1}_m\hat{\vct{u}}^\transpose-\hat{\vct{v}}\vct{1}_n^\transpose,\\
	&=\tilde{C}-\vct{1}_m\hat{\vct{u}}^\transpose-\hat{\vct{v}}\vct{1}_n^\transpose,\\
	&=\tilde{C}_3,\\
	&=\hat{\vec{r}}\hat{\vec{c}}^\transpose.
\end{align*}

Therefore, $(m,n,\tilde{A},\tilde{B})$ is strategically equivalent to $(m,n,\hat{A},\hat{B})$ by Lemma \ref{lem:stratEqVec} and $\hat{A}+\hat{B}=\hat{\vec{r}}\hat{\vec{c}}^\transpose$, which is clearly a rank-$1$ game.
\end{proof}

While Theorem \ref{thm:rank1} handles the most general case of $\rank{D}=2$, for completeness we still need to consider the cases of $\rank{D}<2$. The following theorem addresses those cases.

\subsection{The case of \texorpdfstring{$\rank{\tilde{C}}$}{} less than 3}\label{app:subsec:rankCless3}

\begin{theorem}\label{thm:rank1Less3}
Consider the game  $(m,n,\tilde{A},\tilde{B})$. If there exists $\gamma\in \Re_{>0}$ such that $\tilde{C}:=\tilde{C}(\gamma)=\tilde{A}+\gamma\tilde{B}=D+\hat{\vec{r}}\hat{\vec{c}}^\transpose$ where $D\in\mc{M}_{m\times n}(\Re)$, $\rank{D}<2$, $\rank{\tilde{C}}<3$, $\tilde{C}\notin\mc{M}_{m\times n}(\Re)$, $\hat{\vec{r}}\in\Re^m$, and  $\hat{\vec{c}}\in\Re^n$ then:
\begin{enumerate}
    \item If $\rank{D}=0$, $\hat{A}= \tilde{A}$ and $\hat{B}=\gamma\tilde{B}$.\label{itm:thm:rank1Less3:rank1}
    \item If $\rank{D}=1$ and $\vct{1}_m\in \colspan{\tilde{C}}$ then $\hat{A}=\tilde{A}-\vec{1}_m\hat{\vec{u}}^\transpose$ and $\hat{B}=\gamma\tilde{B}$.\label{itm:thm:rank1Less3:rank2OneM}
    \item If $\rank{D}=1$ and $\vct{1}_n\in \colspan{\tilde{C}^\transpose}$ then $\hat{A}= \tilde{A}$ and $\hat{B}=\gamma\tilde{B}-\hat{\vec{v}}\vec{1}_n^\transpose$.\label{itm:thm:rank1Less3:rank2OneN}
\end{enumerate}
Then the bimatrix game $(m,n,\tilde{A},\tilde{B})$ is strategically equivalent to the rank-$1$ game $(m,n,\hat{A},\hat{B})$.
\end{theorem}
\begin{proof}
For case \ref{itm:thm:rank1Less3:rank1}, $\rank{D}=0$ and $\tilde{C}\notin\mc{M}_{m\times n}(\Re)$ implies that $\tilde{C}=\hat{\vec{r}}\hat{\vec{c}}^\transpose$. The result immediately follows from this observation.

For case \ref{itm:thm:rank1Less3:rank2OneM}, $\tilde{C}\notin\mc{M}_{m\times n}(\Re)$, $\rank{D}=1$, and $\vct{1}_m\in \colspan{\tilde{C}}$ implies that $\tilde{C}=\vct{1}_m\hat{\vct{u}}^\transpose+\hat{\vec{r}}\hat{\vec{c}}^\transpose$. Then, using a proof similar to Theorem \ref{thm:rank1}, it is straightforward to show that after applying the Wedderburn rank reduction formula we have   
\begin{align*}
    	\tilde{C}_2 &= \tilde{C}-w_1^{-1}\tilde{C}\vct{x}_1\vct{y}_1^\transpose \tilde{C}=\tilde{C}-\vct{1}_m\hat{\vct{u}}^\transpose,\\
    	&=\hat{\vec{r}}\hat{\vec{c}}^\transpose.
    \end{align*}
    
Then, with $\hat{A}$ and $\hat{B}$ as defined in case \ref{itm:thm:rank1Less3:rank2OneM}, we have that:
\begin{align*}
	\hat{A}+\hat{B}&=\tilde{A}+\gamma\tilde{B}-\vct{1}_m\hat{\vct{u}}^\transpose,\\
	&=\tilde{C}-\vct{1}_m\hat{\vct{u}}^\transpose,\\
	&=\tilde{C}_2,\\
	&=\hat{\vec{r}}\hat{\vec{c}}^\transpose.
\end{align*}
Therefore, $(m,n,\tilde{A},\tilde{B})$ is strategically equivalent to $(m,n,\hat{A},\hat{B})$ by Lemma \ref{lem:stratEqVec} and $(m,n,\hat{A},\hat{B})$ is a rank-$1$ game. The proof of case \ref{itm:thm:rank1Less3:rank2OneN} is similar to the proof of case \ref{itm:thm:rank1Less3:rank2OneM} and therefore omitted.
\end{proof}

\section{An Efficient Decomposition of \texorpdfstring{$\tilde{C}$}{} -- The Remaining Cases}\label{app:decompC}

In Theorem \ref{thm:rank1GammaUnknown},  we assumed the existence of a decomposition $\tilde{C}(\gamma)=M+D$, where  $D\in\mc{M}_{m\times n}(\Re)$ with $\rank{D}=2$. In this section, we present the decomposition for $\rank{D}<2$. This decomposition follows from Theorem \ref{thm:rank1Less3}, thus there are 3 cases to consider.

\begin{theorem}\label{thm:rank1GammaUnknownLess3}
Consider the game $(m,n,\tilde{A},\tilde{B})$.
 \begin{enumerate}
    	\item If there exists $\gamma^*\in\Re_{>0}$ such that $\rank{\tilde{A}+\gamma^*\tilde{B}}=1$ then $\rank{D}=0$, $\hat{A}= \tilde{A}$ and $\hat{B}=\gamma^*\tilde{B}$.\label{itm:thm:rank1GammaUnknownLess3:rank1}
        \item else; select any $i\in\{1\dots m\}$. Let $\bar{A}+\gamma\bar{B}=\tilde{A}+\gamma\tilde{B}-(\vec{1}_m\tilde{A}_{(i)}+\gamma\vec{1}_m\tilde{B}_{(i)})$. If there exists $\gamma^*\in\Re_{>0}$ such that $\rank{\bar{A}+\gamma^*\bar{B}}=1$ then $\rank{D}=1$, $\vct{1}_m\in \colspan{\tilde{C}}$, $\hat{A}=\tilde{A}-(\vec{1}_m\tilde{A}_{(i)}+\gamma^*\vec{1}_m\tilde{B}_{(i)})$, and $\hat{B}=\gamma^*\tilde{B}$.\label{itm:thm:rank1GammaUnknownLess3:rank2OneM}
        \item  else; select any $j\in\{1\dots n\}$. Let $\bar{A}+\gamma\bar{B}=\tilde{A}+\gamma\tilde{B}-(\tilde{A}^{(j)}\vec{1}_n^\transpose+\gamma\tilde{B}^{(j)}\vec{1}_n^\transpose)$. If there exists $\gamma^*\in\Re_{>0}$ such that $\rank{\bar{A}+\gamma^*\bar{B}}=1$ then $\rank{D}=1$, $\vct{1}_n\in \colspan{\tilde{C}^\transpose}$, $\hat{A}= \tilde{A}$, and $\hat{B}=\gamma^*\tilde{B}-(\tilde{A}^{(j)}\vec{1}_n^\transpose+\gamma^*\tilde{B}^{(j)}\vec{1}_n^\transpose)$.\label{itm:thm:rank1GammaUnknownLess3:rank2OneN}
 \end{enumerate}	
    	 If any of the three cases hold true, then the game $(m,n,\tilde{A},\tilde{B})$ is strategically equivalent to the rank-$1$ game $(m,n,\hat{A},\hat{B})$.   
\end{theorem}

\begin{proof}
 The proof follows from Theorem \ref{thm:rank1Less3} and the discussion preceding Theorem \ref{thm:rank1GammaUnknown}.
\end{proof}

\section{Extending ShortSER1 to SER1}\label{app:algAnalysis} 

Extending \textsc{ShortSER1} to \textsc{SER1} is rather straightforward. One simply has to consider each of the 3 cases from Theorem \ref{thm:rank1GammaUnknownLess3} in order, and then consider the case from Theorem \ref{thm:rank1GammaUnknown}.  If any of the 4 cases is true, then the algorithm calculates the appropriate game $(m,n,\hat{A},\hat{B})$ per the case and terminates. Each case from Theorem \ref{thm:rank1GammaUnknownLess3} is similar to the case from Theorem \ref{thm:rank1GammaUnknown}. Thus, each case can be checked in time $\bigO{(mn+M(\mc{L}))}$. Therefore, considering all 4 cases has time $\bigO{(mn+M(\mc{L}))}$. 

\bibliography{ser1_bib}
\bibliographystyle{ieeetr}
\end{document}